\documentclass[a4paper,UKenglish,cleveref, autoref, thm-restate,authorcolumns]{lipics-v2019}

\usepackage{ebproof,graphicx,xcolor,wrapfig,verbatim} 
\usepackage{listings}
\title{Static Race Detection for RTOS Applications} 

\titlerunning{Data-Race Detection} 

\author{Rishi Tulsyan}{Indian Institute of Science Bangalore, India}{rishitulsyan@iisc.ac.in}{}{}

 \author{Rekha Pai}{Indian Institute of Science Bangalore, India}{rekhapai@iisc.ac.in}{}{}

\author{Deepak D'Souza\footnote{corresponding author}}{Indian Institute of Science Bangalore, India}{deepakd@iisc.ac.in}{}{}

\authorrunning{R. Tulsyan, R. Pai and D. D'Souza} 

\Copyright{Rishi Tulsyan, Rekha Pai, and Deepak D'Souza} 



\supplement{\url{bitbucket.org/rishi2289/static\_race\_detect/}}

\acknowledgements{The second author want to thank University Grants Commission (UGC) India for the DS Kothari Post Doctoral Fellowship EN/17-18/0039}

\nolinenumbers 

\newcommand{\nat}{\ensuremath{\mathbb{N}}} 
\newcommand{\ints}{\ensuremath{\mathbb{Z}}} 
\newcommand{\union} {\cup}

\newcommand{\true}{\ensuremath{\mathit{true}}}
\newcommand{\false}{\ensuremath{\mathit{false}}}

\newcommand{\vars}{\ensuremath{V}}

\newcommand{\disable}{\texttt{disableint}}
\newcommand{\enable}{\texttt{enableint}}
\newcommand{\suspend}{\cmd{suspend}}
\newcommand{\resume}{\cmd{resume}}
\newcommand{\create}{\cmd{create}}
\newcommand{\start}{\texttt{start}}
\newcommand{\myskip}{\texttt{skip}}
\newcommand{\lock}{\cmd{lock}}
\newcommand{\unlock}{\cmd{unlock}}
\newcommand{\suspendsched}{\cmd{suspendsched}}
\newcommand{\resumesched}{\cmd{resumesched}}
\newcommand{\block}{\texttt{block}}
\newcommand{\assume}{\texttt{assume}}
\newcommand{\cmd}[1]{\texttt{#1}}
\newcommand{\command}{\ensuremath{cmd}}

\newcommand{\ent}{\mathit{ent}}
\newcommand{\exit}{\mathit{ext}}
\newcommand{\Instr}{\mathit{inst}}
\newcommand{\locs}{L}	
\newcommand{\instr}{\ensuremath{\iota}}

\newcommand{\threads}{\ensuremath{\mathcal{T}}}
\newcommand{\codes}{\ensuremath{\mathit{F}}}
\newcommand{\locks}{\ensuremath{M}}
\newcommand{\states}{\ensuremath{S}}
\newcommand{\suspended}{\mathcal{S}}
\newcommand{\blocked}{\mathcal{B}}
\newcommand{\ready}{\mathcal{R}}
\newcommand{\pmap}{\ensuremath{\mathcal{P}}}
\newcommand{\acquired}{\mathcal{A}}
\newcommand{\trel}{\ensuremath{\Rightarrow}}
\newcommand{\llbracket}{[\![}
\newcommand{\rrbracket}{]\!]}
\newcommand{\smspace}{\hspace{0.2cm}}
\newcommand{\task}[1]{\textit{task}(#1)}
\newcommand{\isr}[1]{\textit{ISR}(#1)}
\newcommand{\taskfn}{\ensuremath{\mathit{task}}}
\newcommand{\isrfn}{\ensuremath{\mathit{ISR}}}
\newcommand{\type}{\ensuremath{\mathit{type}}}
\newcommand{\undef}{\mathit{undef}}
\newcommand{\susplist}{\ensuremath{\mathit{susplist}}}
\newcommand{\reslist}{\ensuremath{\mathit{reslist}}}

\newcommand{\pc}{\ensuremath{\mathit{pc}}}
\newcommand{\env}{\ensuremath{\phi}}

\newcommand{\id}{\ensuremath{\mathit{id}}}
\newcommand{\sd}{\ensuremath{\mathit{ss}}}

\newcommand{\main}{\ensuremath{\mathit{main}}}

\newcommand{\fmap}{\ensuremath{\mathcal{F}}}

\newcommand{\var}[1]{\emph{#1}}

\newcommand{\oibsymb}{\ensuremath{{\triangleleft\,}}}
\newcommand{\noibsymb}{\ensuremath{{\not\!\triangleleft\,}}}

\newcommand{\etal}{\emph{et al.}}

\newcommand{\tool}{RAPR} 

\begin{document}

\maketitle

\begin{abstract}
We present a static analysis technique for detecting data races in
Real-Time Operating System (RTOS) applications.
These applications are often employed in safety-critical tasks and the
presence of races may lead to erroneous behaviour with serious consequences.
Analyzing these applications is challenging due to the variety of
non-standard synchronization mechanisms they use.
We propose a technique based on the notion of an ``occurs-in-between''
relation between statements. This notion enables us to capture the interplay
of various synchronization mechanisms.
We use a pre-analysis and a small set of not-occurs-in-between patterns to
detect whether two statements may race with each other.
Our experimental evaluation shows that the technique is efficient and
effective in identifying races with high precision.
\end{abstract}

\section{Introduction}
\label{sec:introduction}

Real-Time Operating Systems (RTOSs) are small operating systems or
microkernels that an application programmer uses as a library to
create and manage the execution of multiple tasks or threads.
The programs written by the application programmer are called RTOS
applications and are programs typically written in C
or C++ that are compiled along with the RTOS kernel library and run on
bare metal processors.
Much of embedded software today, ranging from home appliances to
safety-critical systems like industrial
automation systems and flight controller software, are implemented as such
programs.

An RTOS application comprises multiple threads (even if these are
typically run on a single core) and hence they need to protect against
concurrency issues like data races.
Two statements are involved in a data race if they are conflicting
accesses to a shared memory location and can happen ``simultaneously''
or one after another.
Data races can lead to unexpected and erroneous program behaviours,
with serious consequence in safety-critical applications.

While detecting data races is important, doing this for
RTOS applications is a challenging problem.
This is because these programs use a variety of non-standard
synchronization mechanisms like dynamically raising and lowering
proirities, suspending other tasks and the scheduler, flag-based
synchronization, disabling and
enabling interrupts, in addition to the more standard locks and
semaphores.
A look at the ArduPilot flight control software \cite{ardupilot} which
is written in C++ and runs on the ChibiOS RTOS shows several instances
of \emph{each} of these synchronization mechanisms being used.
Standard techniques for race detection like lockset analysis
\cite{VoungJL07} or for priority-ceiling based scheduling and
flag-based synchronization \cite{SchwarzSVLM11,schwarz2014}, or the
disjoint-block approach of \cite{ChopraPD19} for disabling interrupts,
would not be precise enough as they do not handle the first two 
mechanisms mentioned above.
Extending the disjoint-block approach for these synchronization
mechanisms seems difficult.

Instead, in this work we adapt the disjoint-block approach of \cite{ChopraPD19}
to focus on a weaker notion of ``not occuring in between''.
Essentially, a statement $s_2$ does not \emph{occur in between} a statement
$s_1$ if it is not possible for a thread running $s_2$ to preempt a
thread while it is running $s_1$.
If $s_1$ and $s_2$ cannot occur in between each other they also cannot
race.
We identify six patterns or rules that ensure that a statement cannot
occur in between another.
We take the help of a pre-analysis to identify dynamic priority ranges
as well as task suspension information, for each statement in an
application.
Then for each pair of conflicting statements we check if the rules
tell us that they cannot occur in between each other.

We have implemented our analysis for FreeRTOS applications
(FreeRTOS \cite{freertos} is a popular open source RTOS), and analyse
several small benchmarks from the literature as well as a fragment of
the ArduPilot \cite{ardupilot} code, which we translate as a FreeRTOS application.
Our analysis runs in fractions of a second with an overall
precision rate of 73\%.

\section{Overview} 
\label{sec:overview}

We begin with an overview of our technique with an illustrative
example adapted from a 
\begin{wrapfigure}[18]{r}{0.5\textwidth}
\vspace{-0.5cm}
\begin{scriptsize}
\begin{lstlisting}
void main(...) {                Prio  Susp 
1. item = count = 0;              
2. xTaskCreate(prod,...,1, t1); 
3. xTaskCreate(cons,...,1, t2); 
4. vTaskStartScheduler(); 
}

void prod(...) {
10. for( ; ; ) {                 1,1  -
11.   vTaskSuspend(t2);          1,1  -
12.   item = 5;                  1,1  cons
13.   count = count+1;           1,1  cons
14.   vTaskResume(t2);           1,1  cons
15. }                            1,1  -
}

void cons(...) {
20. for( ; ; ) {                 1,1  -
21.   temp = item;               1,1  -   
22.   vTaskPrioritySet(NULL, 2); 1,1  -   
23.   count = count-1;           2,2  -  
24.   vTaskPrioritySet(NULL, 1); 2,2  -
25. }                            1,1  -
}
\end{lstlisting}
\end{scriptsize}
\caption{A producer-consumer FreeRTOS app}
\label{fig:prod-cons-freertos}
\end{wrapfigure}
FreeRTOS demo application.
The application, shown in Fig.~\ref{fig:prod-cons-freertos}, begins by
creating two task threads \texttt{t1} and \texttt{t2} that run the
task functions \texttt{prod} and \texttt{cons} respectively, both at
priority 1.
Once the scheduler is started in line~4 of \texttt{main}, the two
threads begin executing in a round-robin manner, preempting each other
whenever the time slice is over (unless one thread is suspended, or
the running thread has raised its priority above the other thread).
The \texttt{prod} thread protects its accesses to the shared variables
\texttt{item} and \texttt{count} by suspending the \texttt{cons}
thread in line~11, and resuming it in line~14 after the access.
Similarly, the \texttt{cons} thread protects its access to
\texttt{count} by temporarily raising its priority to 2 in line~22.

We are interested in statically detecting potential data races in this
application. We give a more precise definition of a race in
Sec.~\ref{sec:race}, but for now we can take it to mean that two
statements access a shared variable with at least one writing to it
(we call these ``conflicting'' accesses),
and these statements happen one after the other in some execution of
the application.

Our analysis begins by first performing a data-flow analysis to
identify the minimum and maximum dynamic priorities that each
statement can run at.
The computed values are shown on the second column from the right in
the figure, and represent the priorities just before the statement.
Thus at line~23 in \texttt{cons} the min and max priorities are both
2.
We also perform a ``suspended'' analysis to find out at each point,
which are the tasks that are guaranteed to be suspended. These values
are shown in the rightmost column.

Next, for each conflicting pair of accesses $s_1$ and $s_2$, we check
whether $s_2$ can ``occur in between'' $s_1$.
Essentially, $s_2$ can occur in between $s_1$ if there is an execution
in which while $s_1$ is executing, a context-switch may happen and
$s_2$ eventually executes before the context switches back to $s_1$.
If $s_2$ cannot occur in between $s_1$, and vice-versa, then one can
\emph{rule out} $s_1$ and $s_2$ being involved in a race.
To check the ``occur in between'' relation we use a small set of rules
(see Fig.~\ref{fig:occur-in-between} in Sec.~\ref{sec:rules}) which
tell us when $s_2$ \emph{cannot} occur in between $s_1$.
Thus, by the ``Suspend'' rule (C1), we can conclude that statements in
line~21 and 23 cannot occur in between the statements in line~12 and
13 (since the \texttt{cons} task is suspended here).
Similarly, by the ``Priority'' rule (C2), it follows that line~13
cannot occur in between line~23 (since it runs at a higher priority).
This allows us to conclude that the accesses to \texttt{count} in
lines~13 and 23 cannot race.
However for the accesses to \texttt{item} in lines~12 and 21, we are
unable to show that line~12 cannot occur in between 21, and hence our
analysis declares them as potentially racy. Indeed, these two accesses
are racy.

We note that analyses like \cite{SchwarzSVLM11,ChopraPD19} do not
handle these synchronization mechanisms and would be unable to declare
the accesses in line~13 and 23 to be non-racy.

\section{Interrupt-Driven Applications}
\label{sec:language}

In this section, we describe the syntax and semantics of an Interrupt-Driven
Application (IDA).
An IDA program is essentially a set of thread functions, which are
run by dynamically created threads during execution.
The functions are of two types: \emph{task} functions which will be
run by threads that are created dynamically at different priorities,
and \emph{ISR} functions which are run as Interrupt Service Routines
triggered by hardware interrupts, at fixed
priorities above that of task threads.
There is a designated \main\ function which is run by the
\main\ thread which is the only thread running initially.
The \main\ thread may create other task threads and then ``\start''
the scheduler, at which point the created threads and ISR threads are
enabled.
The scheduler runs the task threads according to a
highest-priority-first basis and time-slices within threads of the same
priority. ISR threads can be triggered at any point of time,
preempting task threads or lower priority ISR threads.

 
The thread functions can use a variety of commands, listed in
Tab.~\ref{tab:instr}, to perform computation or influence the way they
are scheduled.
Task threads are created using the \create\ command. The command
creates a new thread, which runs the specified
task function at the specified priority.
High priority threads share execution time 
with low priority threads
using the \cmd{set$\_$priority}, \suspend, and \block\ commands. 
These commands can lead to re-scheduling of the threads, thereby giving
other threads a chance to execute. The \cmd{set$\_$priority} command
sets the priority of 
a task thread, \suspend\ suspends the execution of a task thread, and \block\ 
(representing blocking commands like ``delay'' or ``receive message'')
blocks the execution of the current task thread.
A suspended task thread can be resumed 
with the \resume\ command. A blocked task thread resumes after a
non-deterministic amount of time.
Task threads can suspend and resume the scheduler with 
\suspendsched\ and \resumesched, respectively.
When the scheduler is suspended the currently running task thread
can be preempted only by an ISR thread, and not by other task threads.
Threads can also
disable and enable interrupts with \disable\ and 
\enable, respectively.
When interrupts are disabled, no preemption can occur.
Tasks can synchronize accesses to shared
variables by acquiring and releasing locks with \lock\ and
\unlock\ commands, respectively. 

\begin{figure}
 \begin{subfigure}[b]{.73\linewidth}
  \begin{scriptsize}
   \begin{verbatim}
main:                  prod:                   cons:
1. item:=0;            10. for(; ;) {          20. for(; ;) {
2. count:=0;           11.   suspend(t2);      21.   temp:=item;
3. create(prod,1,t1);  12.   item:=5;          22.   set_priority(t2,2);
4. create(cons,1,t2);  13.   count:=count+1;   23.   count:=count-1;
5. start;              14.   resume(t2);       24.   set_priority(t2,1);
6.                     15. }                   25. }
                       16.                     26.
   \end{verbatim}
  \end{scriptsize}
  \caption{Example IDA program}
 \end{subfigure}%
 \begin{subfigure}[b]{.30\linewidth}
   \includegraphics[scale=0.6]{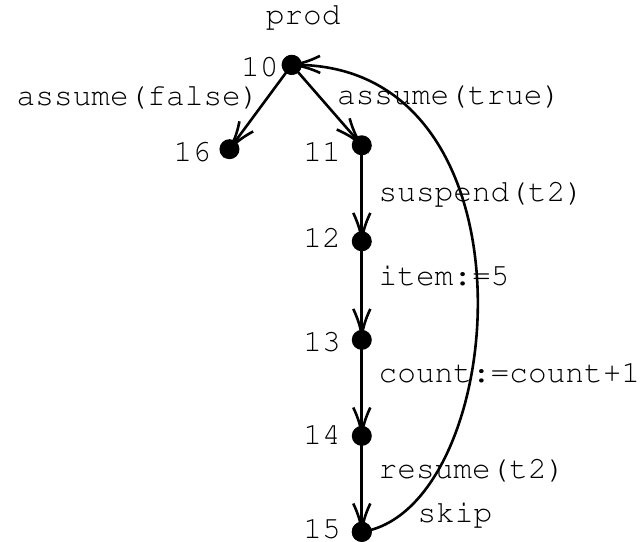}
    \caption{CFG of \texttt{prod} }
 \end{subfigure}
\caption{Example program and the CFG representation of \texttt{prod}}\label{fig:IDA-ex}
\end{figure}

\label{sub:syn}
More formally, an IDA program $P$ is a triple
$\langle \vars, \locks, \codes \rangle$ 
where $\vars$ is a finite set of integer-valued global variables,
$\locks$ is a finite set of locks, and $\codes$ is a finite set of
thread function names, with a designated one called \main.
Each function $A$ in $\codes$ has an associated Control Flow Graph
(CFG) $G_{A}=(\locs_{A}, \ent_A, \exit_A, \Instr_{A})$, 
where $\locs_{A}$ is the set of \emph{locations},
$\ent_A$ and $\exit_A$ are respectively the \emph{entry} and \emph{exit}
locations in $\locs_A$, and
$\Instr_{A} \subseteq \locs_{A} \times
\command(\vars,\locks) \times
\locs_{A}$ is the set of \emph{instructions} of the CFG.
Here $\command(\vars,\locks)$ is the set of commands in
Tab.~\ref{tab:instr} over the
variables $\vars$ and locks $\locks$.
Each function $A$ in $\codes$ also has an associated \emph{type},
$\type(A)$, which is one of \taskfn\ or \isrfn.
While task threads are created during execution at priorities
specified in the \create\ command, ISR threads run at a fixed static
priority.
We assume that during execution task threads can have priorities upto a
constant value $m \in \nat$ (which we fix for all IDA programs),
while ISR threads have distinct
priorities which are greater than $m$.
If $\{f_1, \ldots, f_k\}$ are the functions of type \isrfn, then
without loss of generality we assume their priorities to be $m+1,
\ldots, m+k$ respectively.
The IDA version of the FreeRTOS application from
Fig.~\ref{fig:prod-cons-freertos} is shown in Fig.~\ref{fig:IDA-ex}.

Some notation will be useful going forward. For 
a program $P$, the instructions of $P$, denoted $\Instr_{P}$, is the union of 
instructions in the thread functions of $P$, and locations in $P$,
denoted $\locs_{P}$, 
is the union of the locations in the thread functions of $P$.
An IDA program allows
standard integer and Boolean expressions over $\vars$.
For an integer expression $e$, Boolean expression $b$, and an environment 
$\env$ for $\vars$, $\llbracket e \rrbracket_{\env}$ denotes the integer
value that $e$ evaluates to in $\env$, and $\llbracket b \rrbracket_{\env}$
denotes the Boolean value that $b$ evaluates to in $\env$. For a map 
$f: X \rightarrow Y$ and elements $x,y$ which may or may not be in
$X$ or $Y$, we use the notation $f[x \mapsto y]$ to denote
the function $f':X \union \{x\} \rightarrow Y \union \{y\}$ given by
$f'(x)=y$ and for all $z$ different from $x$, $f'(z) = f(z)$.

\begin{table}[h]
  \caption{IDA Basic Commands}\label{tab:instr}
  \begin{small}
    \begin{tabular}[c]{|l|l|}
     \hline
        \textbf{Command} & \textbf{Description} \\ \hline
        \myskip & Do nothing. \\ \hline        

        \var{x} := \var{e} & Assign the value of expression \var{e} to variable \var{x}. \\ \hline
        
        \assume(\var{b}) & Enabled only if expression \var{b} evaluates to 
                                \textit{true}; does nothing. \\ \hline                              
        
        \create($A, p, t$) & Create task thread with func $A$,
        prio $p$, and store thread id in variable $t$. \\ \hline
        
        \cmd{set$\_$priority}($t, p$) & Set priority of task thread
        $t$ to $p$. When the first parameter is NULL,  \\
        & set priority of current thread. Allowed only in task function. \\  \hline
        
        \suspend($t$) & Suspend task thread $t$. When the parameter is 
        NULL, suspend \\ 
        & current thread. Allowed only in task function. \\ \hline
        
        \resume($t$) & Resume task thread $t$.  Allowed only in task
        function. \\ \hline
        
        \suspendsched & Suspend scheduler. Disables switching to other
        task threads. \\ \hline
        
        \resumesched & Resume the scheduler. Enables switching to other task
        threads. \\ \hline
        
        \disable & Disable interrupts and suspend the scheduler. \\ \hline
        
        \enable & Enable interrupts and resume the scheduler. \\ \hline
        
        \lock($l$) & Acquire lock $l$. Blocks if $l$ is not available. \\ \hline                                  
        \unlock($l$) & Release lock $l$. \\ \hline
        
        \block & Block the current task thread. Re-enable after
        non-deterministic delay. \\ \hline
                    
        \start & Start scheduler and enable interrupts. Called only by \main. \\ \hline                
     \end{tabular}
  \end{small}
\end{table}

\label{sub:sem}
We can now define the semantics of an IDA program $P = \langle \vars,
\locks, \codes \rangle$ as a
labeled transition system
$\langle\states, \Sigma, \trel, s_{0} \rangle$,
whose components are defined as follows.
Let $f_1, \ldots, f_k$ be the thread functions of type \isrfn, with
priorities $m+1, \ldots, m+k$ respectively.

The set of states $\states$ contains tuples of the form $s = \langle \blocked, \suspended, \ready, \pmap, \acquired, \fmap, \pc, \env, r, i, \sd, \id \rangle$, where 
\begin{itemize}
\item $\blocked$, $\suspended$, and $\ready$ are sets of thread ids
  (which we assume to be simply integers)
  representing the set of \emph{blocked} task threads, \emph{suspended}
  task threads, and \emph{ready} task threads, respectively.
  The sets $\blocked$, $\suspended$, and $\ready$ are pairwise
  disjoint.
  We denote the set of threads created so far by $\threads = \blocked \union
       \suspended \union \ready$.
       
 \item $\pmap: \threads \rightarrow \nat$ gives the
       current priority of each thread.
       
 \item $\acquired: \locks \rightharpoonup \threads$ is a partial map giving
       us the thread that has acquired a particular lock.
       
 \item $\fmap: \threads \rightarrow \codes$ gives the function 
       associated with a thread.
       
 \item $\pc: \threads \rightarrow \locs_{P}$ gives the
   current location of a thread $t$ in the CFG of $\fmap(t)$.
       
 \item $\env \in \vars \rightarrow \ints$ is a valuation for the variables.
 
 \item $r \in \ready$ is the currently running thread, while 
       $i \in \ready$ is the interrupted task thread.
       
 \item $\sd$ is a Boolean value indicating whether the scheduler is 
       suspended ($\sd=\true$) or not, while 
       $\id$ is a Boolean value indicating whether interrupts are
       disabled ($\id=\true$) or not.

\end{itemize}
The initial state $s_{0}$ is $\langle 
 \emptyset, \ 
 \{1,\ldots, k\}, \ 
 \{0\}, \ 
 \{0 \mapsto 0, 1 \mapsto m+1, \ldots, k \mapsto m+k \}, \  
 \emptyset, \ $\\
 $\{0 \mapsto \main, 1 \mapsto f_1, \ldots, k \mapsto f_k\}, 
 \lambda t \in \threads. \ent_{\fmap(t)}, \ 
 \lambda x \in V. 0,\ 
 0,\ 
 0,\ 
 \true,\ 
 \true 
 \rangle$.
Thus initially, no threads are blocked, ISR
threads $1, \ldots, k$ (with priorities $k+1, \ldots, k+n$
respectively) are disabled, and the main thread $0$ with priority $0$ is
ready and also running. No locks are acquired. The threads 
$0,1,\ldots, k$ are associated with their functions.
All the threads are at their entry locations
and all variables are initialized to zero. The interrupted thread is taken 
to be $0$, the scheduler is suspended, and interrupts are disabled.

The transition relation $\trel$ is given as follows. 
Consider a state $s$ expressed as the tuple $s = \langle \blocked, \suspended, \ready, \pmap,
\acquired, \fmap, \pc, \env, r, i, \sd, \id \rangle$,
a thread $t \in \threads$,
and an instruction $\instr = (l, c, l')$ in $\fmap(t)$.
Then we have $s \trel_{\instr} s'$
iff one of the following rules is satisfied.
Each rule says that if the conditions on command $c$ and state $s$, specified
in the antecedent of a rule (above the line), hold then $s \trel_{\instr} s'$,
specified in the consequent of the rule (below the line), holds.
We use $\task{t}$ to indicate that $t$ is a task thread (i.e.\@
($\type(t) = \taskfn$) and $\isr{t}$ to indicate that $t$ is an ISR thread.

In the interest of space, only few rules are shown here. The full semantics
can be found in Arxiv. The ASSIGN
is a simple rule on assignment statement. The ASSIGN-INT rule shows how 
interrupts are handled while CREATE-CS and CREATE-NS rules show how the 
execution of a statement can lead to context switch and no switch, 
respectively, and the START rule shows how the threads 
get running. 
For the 
ASSIGN-INT rule given below, the condition $\pc(t)=l=\ent_{ \fmap(t)}$ should 
hold while for others $\pc(t)=l$ needs to be true.
\\[0.5cm]
\scalebox{0.8}{
\[
\hspace*{-0.5cm}
 \begin{prooftree}
  \Hypo{c=\var{x}:=\var{e} \indent t=r}
  
  \Infer1[\tiny ASSIGN]{s \trel_{\instr} \langle \blocked, \suspended, \ready, \pmap, \acquired, \fmap, \pc[t \mapsto l'], \env[x \mapsto \llbracket e \rrbracket_{\env} ], r, i, \sd, \id \rangle}
 \end{prooftree}
\]
}
\\[0.5cm]
\scalebox{0.8}
{
\[
 \hspace*{-0.5cm}
 \begin{prooftree}
  \Hypo{c=\var{x}:=\var{e} \indent t \in \ready \indent \isr{t} \indent t \neq r \indent \pmap(t) > \pmap(r) \indent \id = \false} 
  
  \Infer1[\tiny ASSIGN-INT]{s \trel_{\instr} \langle \blocked, \suspended, \ready, \pmap, \acquired, \fmap, \pc[t \mapsto l'], \env[x \mapsto \llbracket e \rrbracket_{\env} ], t, r, \sd, \id \rangle}
 \end{prooftree}
\]
}
\\[0.5cm]
\scalebox{0.8}
{
\[
 \begin{prooftree}
  \Hypo{ { c=\create(\var{A}, \var{p}, \var{v}) \smspace t=r \smspace
      \task{t} \smspace A \in \codes \smspace \type(A) = \taskfn \smspace ts \notin \threads \smspace (p \leq \pmap(r)  \vee  (\sd \vee \id) = \true) }} 
  
  \Infer1[\tiny CREATE-NS]{s \trel_{\instr} \langle \blocked,
    \suspended, \ready \union \{ts\}, \pmap[ts \mapsto p], \acquired,
    \fmap[ts \mapsto A], \pc[t \mapsto l', ts \mapsto \ent_{A}], \env[v \mapsto ts], r, i, \sd, \id \rangle}
 \end{prooftree}
\]
}
\\[0.5cm]
\scalebox{0.8}{
\[
 \begin{prooftree}
  \Hypo{{ c=\create(\var{A}, \var{p}, \var{v}) \smspace t=r \smspace \task{t} \smspace A \in \codes \smspace \type(A) = \taskfn \smspace  ts \notin \threads \smspace p > \pmap(r) \smspace (\sd \vee \id)=\false  }}
  
  \Infer1[\tiny CREATE-CS]{s \trel_{\instr} \langle \blocked, \suspended, \ready \union \{ts\}, \pmap[ts \mapsto p], \acquired, \fmap[ts \mapsto A], \pc[t \mapsto l', ts \mapsto \ent_{A}], \env[\var{v} \mapsto ts], ts, i, \sd, \id \rangle}
 \end{prooftree}
\]
}
\\[0.5cm]
\scalebox{0.8}{
\[
 \begin{prooftree}
  \Hypo{{ c=\start \smspace t=r=0 \smspace  (\sd \vee \id)=\false \smspace \exists ts \in (\suspended \union \ready). \task{ts} \ \wedge \ \pmap(ts) = max(\{\pmap(u) | u \in \suspended \union \ready \ \wedge \ \task{u} \}) }}
  
  \Infer1[\tiny START]{s \trel_{\instr} \langle \blocked, \emptyset, \suspended \union \ready, \pmap, \acquired, \fmap, \pc[t \mapsto l'], \env, ts, i, \false, \false \rangle }
 \end{prooftree}
\]
}
\\[0.5cm]

An \emph{execution} $\sigma$ of $P$ is a finite sequence of transitions in the 
transition system defined.
$\sigma = \tau_{0}, \tau_{1}, \cdots, \tau_{n}$, where $n \geq 0$ and
there exists a finite sequence of states $s_{0}, s_{1}, \cdots, s_{n+1}$
in $\states$ such that $s_{0}$ is the initial state and
$\tau_{i} = s_{i} \trel s_{i+1}$ for each $0 \leq i \leq n$.

\section{Data Races and the Occur-in-Between Relation}
\label{sec:race}

We use the notion of data races introduced by Chopra et al
\cite{ChopraPD19}, which is a general notion that applies to programs
with non-standard synchronization mechanisms. The definition
essentially says that two statements
in a program race if (a) they are conflicting accesses to a memory
location and (b) they may happen in parallel, in that notional
``\myskip\ blocks'' around these statements overlap with each other in
some execution of the program.
The definition is
\begin{wrapfigure}[6]{r}{0.3\textwidth}
  \begin{picture}(0,0)%
\includegraphics{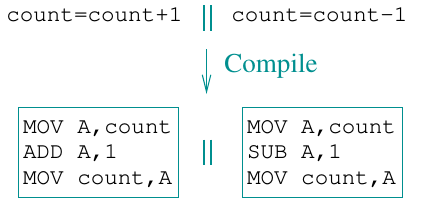}%
\end{picture}%
\setlength{\unitlength}{2693sp}%
\begingroup\makeatletter\ifx\SetFigFont\undefined%
\gdef\SetFigFont#1#2#3#4#5{%
  \reset@font\fontsize{#1}{#2pt}%
  \fontfamily{#3}\fontseries{#4}\fontshape{#5}%
  \selectfont}%
\fi\endgroup%
\begin{picture}(2984,1512)(2254,-3248)
\end{picture}%

\end{wrapfigure}
meant to capture the fact that when these two
statements are compiled down to instructions of a processor, the
interleaving of these instructions may lead to undesirable behaviours
of the program which don't correspond to any sequential execution of
the two statements.
For example in the figure alongside, the conflicting accesses to
\texttt{count} may get compiled to the instructions shown, and the
interleaving of these two blocks of instructions may lead to
unexpected results like \texttt{count} getting decreased by 1 despite
both blocks having completed.

We now define these notions more formally in our setting.
Let us fix an IDA program $P$.
Let $s_1$ and $s_2$ be two instructions in $P$, with associated
commands $c_1$ and $c_2$ respectively.
We restrict ourselves to the case where $c_1$ and $c_2$ are assignment
or assume statements.
We say $s_1$ and $s_2$ are \emph{conflicting accesses} to a variable $x$ if
they both access $x$ and at least one of them writes $x$.
Let $P_{s_1}$ denote the program
obtained from $P$ by 
inserting \myskip\ statements immediately before and after $s_1$.
Similarly, let $P_{s_1,s_2}$ denote the program
obtained from $P$ by 
inserting \myskip\ statements immediately before and after both $s_1$ and
$s_2$.
We say $s_1$ and $s_2$ \emph{may happen in parallel} (MHP) in $P$ if
there is an execution of $P_{s_1,s_2}$ in which the two \myskip-blocks
interleave (i.e.\@ one block begins in between the other).
These terms are illustrated in Fig.~\ref{fig:race-def}. We use the
convention that $A$ and $B$ represent the static thread functions,
while $t_A$ and $t_B$ represent dynamic threads that run the functions
$A$ and $B$ respectively, with an optional subscript indicating the
priority at which the thread was created.
Finally we say $s_1$ and $s_2$ are involved in a \emph{data-race} (or simply
are \emph{racy}) in $P$, if
they are conflicting accesses that may happen in parallel in $P$.

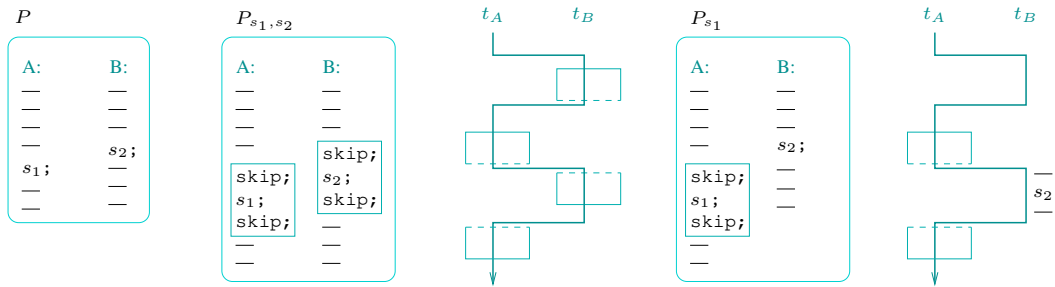
\begin{figure}
  \begin{center}
    \begin{picture}(0,0)%
\includegraphics{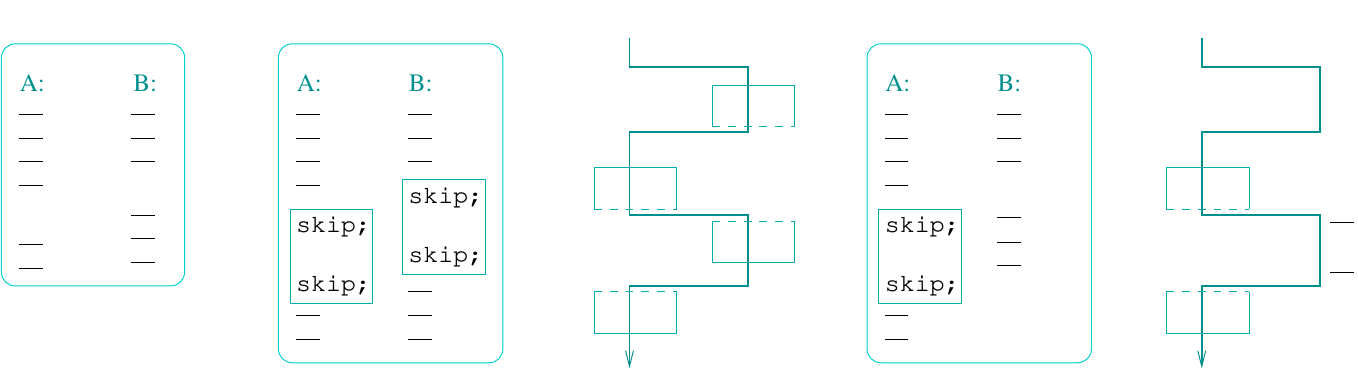}%
\end{picture}%
\setlength{\unitlength}{2486sp}%
\begingroup\makeatletter\ifx\SetFigFont\undefined%
\gdef\SetFigFont#1#2#3#4#5{%
  \reset@font\fontsize{#1}{#2pt}%
  \fontfamily{#3}\fontseries{#4}\fontshape{#5}%
  \selectfont}%
\fi\endgroup%
\begin{picture}(10329,2812)(1479,-2603)
\put(3736, 29){\makebox(0,0)[lb]{\smash{{\SetFigFont{7}{8.4}{\rmdefault}{\mddefault}{\updefault}{\color[rgb]{0,0,0}$P_{s_1,s_2}$}%
}}}}
\put(3736,-1771){\makebox(0,0)[lb]{\smash{{\SetFigFont{7}{8.4}{\ttdefault}{\mddefault}{\updefault}{\color[rgb]{0,0,0}$s_1$;}%
}}}}
\put(4591,-1546){\makebox(0,0)[lb]{\smash{{\SetFigFont{7}{8.4}{\ttdefault}{\mddefault}{\updefault}{\color[rgb]{0,0,0}$s_2$;}%
}}}}
\put(1626,-1456){\makebox(0,0)[lb]{\smash{{\SetFigFont{7}{8.4}{\ttdefault}{\mddefault}{\updefault}{\color[rgb]{0,0,0}$s_1$;}%
}}}}
\put(1556, 39){\makebox(0,0)[lb]{\smash{{\SetFigFont{7}{8.4}{\rmdefault}{\mddefault}{\updefault}{\color[rgb]{0,0,0}$P$}%
}}}}
\put(6276, 74){\makebox(0,0)[b]{\smash{{\SetFigFont{7}{8.4}{\rmdefault}{\mddefault}{\updefault}{\color[rgb]{0,.56,.56}$t_A$}%
}}}}
\put(7176, 74){\makebox(0,0)[b]{\smash{{\SetFigFont{7}{8.4}{\rmdefault}{\mddefault}{\updefault}{\color[rgb]{0,.56,.56}$t_B$}%
}}}}
\put(10636, 74){\makebox(0,0)[b]{\smash{{\SetFigFont{7}{8.4}{\rmdefault}{\mddefault}{\updefault}{\color[rgb]{0,.56,.56}$t_A$}%
}}}}
\put(11536, 74){\makebox(0,0)[b]{\smash{{\SetFigFont{7}{8.4}{\rmdefault}{\mddefault}{\updefault}{\color[rgb]{0,.56,.56}$t_B$}%
}}}}
\put(8221, 29){\makebox(0,0)[lb]{\smash{{\SetFigFont{7}{8.4}{\rmdefault}{\mddefault}{\updefault}{\color[rgb]{0,0,0}$P_{s_1}$}%
}}}}
\put(8221,-1771){\makebox(0,0)[lb]{\smash{{\SetFigFont{7}{8.4}{\ttdefault}{\mddefault}{\updefault}{\color[rgb]{0,0,0}$s_1$;}%
}}}}
\put(9076,-1226){\makebox(0,0)[lb]{\smash{{\SetFigFont{7}{8.4}{\ttdefault}{\mddefault}{\updefault}{\color[rgb]{0,0,0}$s_2$;}%
}}}}
\put(2476,-1276){\makebox(0,0)[lb]{\smash{{\SetFigFont{7}{8.4}{\ttdefault}{\mddefault}{\updefault}{\color[rgb]{0,0,0}$s_2$;}%
}}}}
\put(11606,-1696){\makebox(0,0)[lb]{\smash{{\SetFigFont{7}{8.4}{\ttdefault}{\mddefault}{\updefault}{\color[rgb]{0,0,0}$s_2$}%
}}}}
\end{picture}%

  \caption{A program $P$; its transformation $P_{s_1,s_2}$; an
    execution of $P_{s_1,s_2}$ in which the skip blocks overlap and
    witnesses that $s_1$ and $s_2$ MHP in $P$; the
    program $P_{s_1}$; and an execution of $P_{s_1}$ which witnesses
    occurrence of $s_2$ in between $s_1$.}
  \label{fig:race-def}
  \end{center}
\end{figure}

It will be convenient for us to use a stronger notion than MHP called
``occurs-in-between'' while reasoning about IDA programs.
Once again, if $s_1$ and $s_2$ are statements in $P$, we say that
$s_2$ can \emph{occur-in-between} $s_1$ if there is an execution of
$P_{s_1}$ in which $s_2$ occurs sometime between
the first \myskip\ and the second \myskip\ around $s_1$.
In this case we write $s_1 \oibsymb s_2$, and $s_1 \noibsymb s_2$
otherwise.
The definition of $s_1 \oibsymb s_2$ is
illustrated in the right side of Fig.~\ref{fig:race-def}.
While it is immediate that if $s_2$ occurs in between $s_1$ then they
also MHP, a weaker version of the converse is also true:

\begin{proposition}
  \label{prop:oib-mhp}
Let $s_1$ and $s_2$ be two statements in an IDA program $P$. Then
$s_1$ MHP $s_2$ iff either $s_1$ occurs in between $s_2$ or $s_2$
occurs in between $s_1$.
\qed
\end{proposition}

Thus to conclude that $s_1$ and $s_2$ cannot MHP (and hence not race)
it is enough to show that $s_1$ and $s_2$ cannot occur in between
each other.

\section{Occur-In-Between Rules}
\label{sec:rules}

In this section we focus on statically computing a conservative (i.e.\@ under-)
approximation of the \emph{cannot}-occur-in-between relation for an
IDA program, by giving rules for identifying this
\begin{wrapfigure}[6]{r}{0.35\textwidth}
  \begin{picture}(0,0)%
\includegraphics{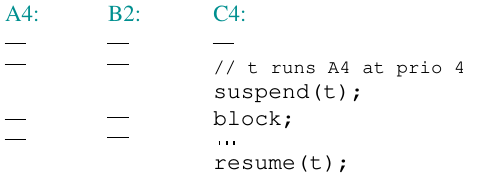}%
\end{picture}%
\setlength{\unitlength}{2279sp}%
\begingroup\makeatletter\ifx\SetFigFont\undefined%
\gdef\SetFigFont#1#2#3#4#5{%
  \reset@font\fontsize{#1}{#2pt}%
  \fontfamily{#3}\fontseries{#4}\fontshape{#5}%
  \selectfont}%
\fi\endgroup%
\begin{picture}(4133,1459)(8184,-1760)
\put(8221,-1091){\makebox(0,0)[lb]{\smash{{\SetFigFont{7}{8.4}{\ttdefault}{\mddefault}{\updefault}{\color[rgb]{0,0,0}$s_1$;}%
}}}}
\put(9076,-1106){\makebox(0,0)[lb]{\smash{{\SetFigFont{7}{8.4}{\ttdefault}{\mddefault}{\updefault}{\color[rgb]{0,0,0}$s_2$;}%
}}}}
\end{picture}%

\end{wrapfigure}
relation.
To illustrate the typical issues we need to keep in mind while framing
these rules, consider the example program alongside.
Task threads $A4$, $B2$, and $C4$ are created at priority 4,2, and 4
respectively.
In the normal course statement $s_2$ in $B2$ would not be able to
occur in between $s_1$ in $A4$ as $A4$ runs at a higher priority than
$B2$. However, (the thread that runs) $C4$ may suspend $A4$ just
before it executes $s_1$, block itself, and allow $B2$ to run.
Thus $s_2$ can occur in between $s_1$.

We will make use of the following
terminology for an IDA program $P$.
Let $s$ be a statement in thread function $A$ in $P$.
We say $s$ may run at priority $p$ if there is an execution of $P$ in
which a thread $t$ runs $A$ and executes statement $s$ at a priority
of $p$.
We say $(p,q)$ is the dynamic priority of $s$ if $p$ and $q$ are
respectively the
minimum and maximum priorities that $s$ can run at.
Similarly, we say that the dynamic priority of a thread function $A$
(or a block of code in $A$) is $(p,q)$ if $p$ and $q$ are respectively
the minimum and maximum priorities at which any statement in $A$ (or
the block of $A$) can run.
Finally, we say that a task function $A$ \emph{may suspend} another task
function $B$ in $P$, if $A$ contains a statement of the form
\suspend$(t)$, and there is an execution of $P$ in which the statement
is executed when the thread
id $t$ points to task function $B$ (that is $t$ runs task $B$).
We say the statement $\suspend(t)$ in $P$ \emph{must suspend} (or
simply \emph{suspends}) a task
$B$ if $t$ takes on a unique thread id at this point along any
execution of $P$, and this thread id is the only thread that runs $B$.
In this case, we will denote such a statement by \suspend$(B)$.

We now proceed to propose sufficient conditions under which one
statement in an IDA program cannot occur-in-between another statement
in the program.
Let us fix an IDA program $P$. 
Let $s_1$ and $s_2$ be statements in thread functions
$A$ and $B$ respectively ($A$ and
$B$ could be the same thread function).
The following conditions (C1)--(C6) below are meant to be sufficient
conditions that ensure that $s_2$ cannot occur in between $s_1$.
In the rules below, by a statement $s$ in a thread function $A$
being enclosed in a
\suspend-\resume\ block we mean there is a path in the CFG of $A$
which contains $s$, begins with a \suspend,
ends with a \resume, and has \emph{no}
intervening \resume\ statement; and similarly for other kinds of
blocks.
Each of these rules is illustrated in Fig.~\ref{fig:occur-in-between}.
\begin{itemize}
\item \textbf{C1} (Suspend Task): Each of the following conditions must hold:
  \begin{itemize}
  \item $s_1$ is enclosed in a $\suspend(B)$-$\resume(B)$
  block with dynamic priority $(p,q)$; 
  \item there is no task with maximum dynamic priority greater than
    or equal to $p$, that can resume $B$;
  \item Either no blocking statement in the $\suspend(B)$-$\resume(B)$
    block, or no other task that can resume $B$.
  \end{itemize}

\item \textbf{C2} (Priority): Each of the conditions below must hold:
  \begin{itemize}
  \item The dynamic priorities of $s_1$ and $s_2$ are $(p_1,q_1)$ and
    $(p_2,q_2)$ respectively, with $p_1 > q_2$.
  \item There is no thread body with maximum dynamic priority greater than or
    equal to $p_1$ that can suspend $A$.
  \end{itemize}

\item \textbf{C3} (Flag): Each of the conditions below must hold:
  \begin{itemize}
  \item $s_1$ is enclosed in a block $F$ begining with setting the
    variable \texttt{flag} to 1 and ending with resetting it to 0,
    with dynamic priority of the block being $(p_1,q_1)$.
  \item The block $F$ is either in the scope of a \suspendsched\ command or there is no thread of priority $\geq p_1$ that resets
    \texttt{flag}.
  \item Either there is no blocking command before $s_1$ in $F$, or no
    thread that can reset flag.
  \item $s_2$ is in an \texttt{if-then} block which checks that
    \texttt{flag} is not set, with the block having dynamic priority
    $(p_2,q_2)$.
  \item $q_1 < p_2$.
  \end{itemize}


\item \textbf{C4} (Lock): Each of $s_1$ and $s_2$ are within a
  \lock$(l)$-\unlock$(l)$ block, for some common lock $l$.

\item \textbf{C5} (Disable Interrupts): $s_1$ is within a
  \disable-\enable\ block.

\item \textbf{C6} (Suspend Scheduler): $s_1$ is within a
  \suspendsched-\resumesched\ block in a task function, and $s_2$ is
  in a task function.
\end{itemize}

\begin{theorem}
\label{thm:noib-soundness}
Let $P$ be an IDA program, and let $s_1$ and $s_2$ be statements in $P$
that satisfy one of the conditions (C1) to (C6) above. Then $s_1
\noibsymb s_2$ in $P$. 
\end{theorem}
\begin{proof}
We sketch here a proof of Thm.~\ref{thm:noib-soundness} on the
soundness of the conditions C1--C6.
Let $P$ be an IDA program with statements $s_1$ and $s_2$ satisfying
one of the conditions C1--C6. We need to argue that in each case $s_1
\noibsymb s_2$.
We focus on the first three rules C1--C3 since the remaining are more
standard and their soundness is easy to see.

\textbf{C1:}
Suppose $s_1$ and $s_2$ satisfy the condition C1, and suppose there is
an execution $\rho$ of $P$ in which $s_2$ occurs in between $s_1$.
Let us say $s_1$ is executed by thread $t_1$ and $s_2$ by thread $t_2$.
Then $s_2$ must happen some time after $t_2$ was suspended by
$t_1$, and before $s_1$ takes place.
The only way this can happen is if:
\begin{itemize}
\item Some thread $t_3$ with priority
\emph{greater than or equal to} $p_1$ resumes $t_2$.
But this is not possible since the condition says that there is \emph{no} other
task with dynamic priority greater than or equal to $p_1$ which can
resume $B$.
    \item $t_1$ makes a blocking call and another task runs and
      resumes $t_2$. However this is ruled out by the requirement that
      [there is no \block\ command before $s_{1}$] OR [there is no
        task other than $A$ which can resume $B$].
\end{itemize}

\smallskip
\textbf{C2:}
Suppose $s_1$ and $s_2$ satisfy the condition C2, and suppose there is
an execution $\rho$ of $P$ in which $s_2$ occurs in between $s_1$.
Let us say $s_1$ is executed by thread $t_1$ and $s_2$ by thread $t_2$.
Then thread $t_2$ must preempt thread $t_1$ during the execution of $s_1$.
The only way this can happen is if:
\begin{itemize}
\item $t_2$ with priority
\emph{greater than} $p_1$ was blocked. It runs and preempts $t_1$.
But this is not possible since the condition says that $p_{1}$ $>$ maximum dynamic priority of $t_2$.
\item $t_2$ has a priority
\emph{equal to} $p_1$ and $t_1$'s time slice expires and it gets
preempted by $t_2$.
Again, this is not possible since the condition says that $p_{1}$ $>$ maximum dynamic priority of $t_2$.
\item Some thread $t_3$ with priority
\emph{greater than or equal to} $p_1$ was blocked. It runs and suspends $t_1$. However this is ruled out by the requirement that
there is no task other than $t_1$ with maximum dynamic priority $\geq$ $p_1$, which can suspend $t_1$.
\end{itemize}

\smallskip
\textbf{C3:}
Suppose $s_1$ and $s_2$ satisfy the condition C3, and suppose there is
an execution $\rho$ of $P$ in which $s_2$ occurs in between $s_1$.
Let us say $s_1$ is executed by thread $t_1$ and $s_2$ by thread $t_2$.
Then $s_2$ must happen some time after $flag_1$ is set to 1 by $t_1$, and before $s_1$ takes place.
The only way this can happen is if:
\begin{itemize}
\item Some thread $t_3$ with priority
\emph{greater than or equal to} $p_1$ was blocked. It runs and resets $flag_1$ to 0. But this is not possible since the condition says that 
$s_1$ is either in the scope of a \suspendsched\ command or there is no thread of priority $\geq$ $p_1$ that resets $flag_1$.
\item $t_1$ makes a blocking call and another task runs and resets $flag_1$ to 0.
Again, this is not possible because of the requirement that [there is no \block\ command before $s_1$] OR [there is no task other than $t_1$ which can reset $flag_1$ to 0].
\item Both $t_1$ and $t_2$ run at the same priority. Before $t_1$ sets $flag_1$ to 1, $t_2$ checks $flag_1$ and finds that it is 0, and enters the block containing $s_2$. Before $t_2$ executes $s_2$, it's time slice expires. It gets preempted by $t_1$ which sets $flag_1$ to 1 and starts $s_1$. However this is ruled out by the requirement that $p_2$ $>$  $p_1$.
\end{itemize}

This completes the argument.
\end{proof}

\begin{figure}
\vspace{-0.6cm}
\begin{picture}(0,0)%
\includegraphics{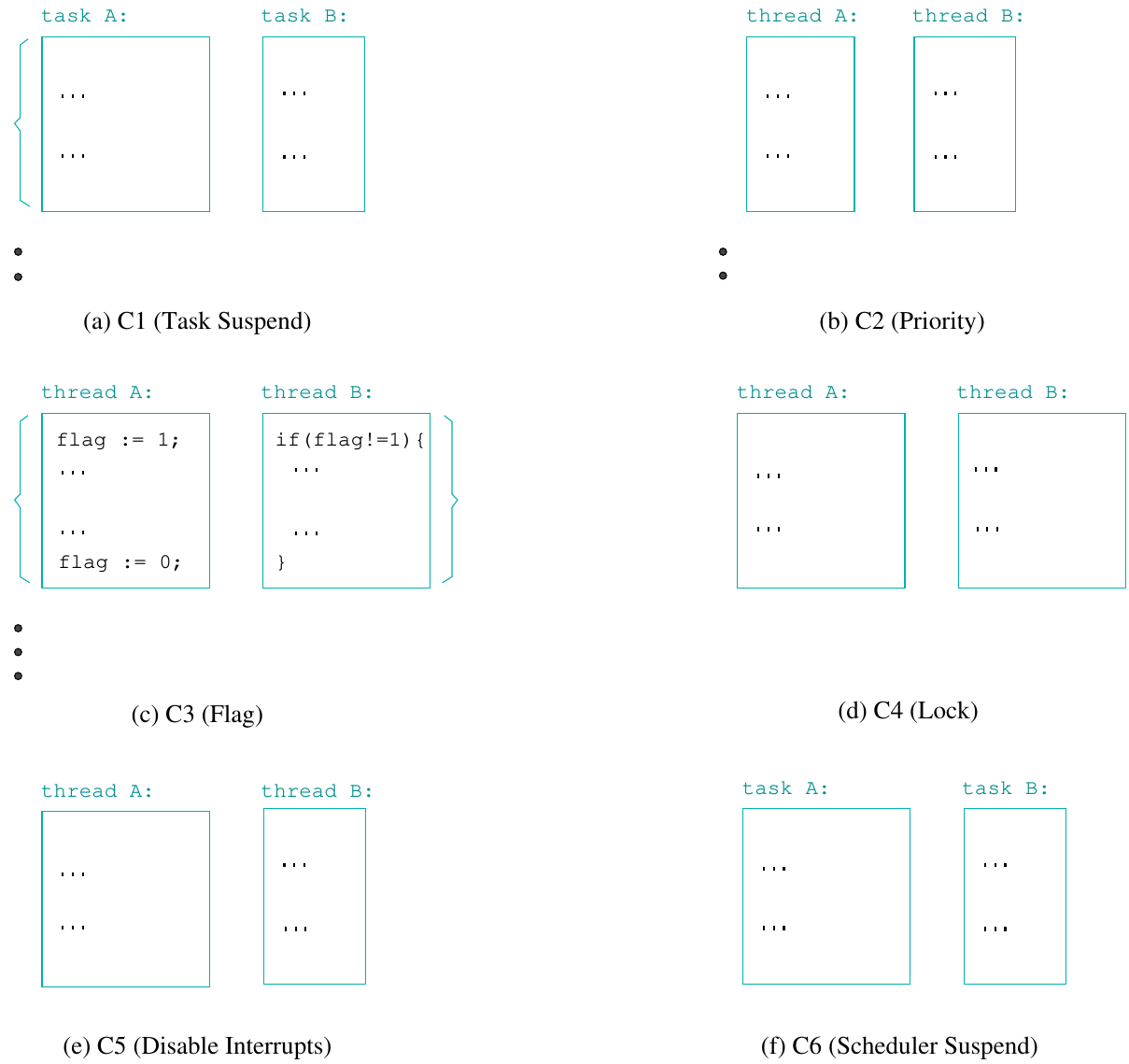}%
\end{picture}%
\setlength{\unitlength}{2693sp}%
\begingroup\makeatletter\ifx\SetFigFont\undefined%
\gdef\SetFigFont#1#2#3#4#5{%
  \reset@font\fontsize{#1}{#2pt}%
  \fontfamily{#3}\fontseries{#4}\fontshape{#5}%
  \selectfont}%
\fi\endgroup%
\begin{picture}(8494,8008)(1034,-7840)
\put(6776,-826){\makebox(0,0)[lb]{\smash{{\SetFigFont{8}{9.6}{\familydefault}{\mddefault}{\updefault}{\color[rgb]{0,0,0}$s_1$;}%
}}}}
\put(7651,-818){\makebox(0,0)[b]{\smash{{\SetFigFont{8}{9.6}{\familydefault}{\mddefault}{\updefault}{\color[rgb]{0,0,0}$\noibsymb$}%
}}}}
\put(8776,-830){\makebox(0,0)[lb]{\smash{{\SetFigFont{8}{9.6}{\sfdefault}{\mddefault}{\updefault}{\color[rgb]{0,.56,.56}$(p_2,q_2)$}%
}}}}
\put(6616,-1771){\makebox(0,0)[lb]{\smash{{\SetFigFont{8}{9.6}{\rmdefault}{\mddefault}{\itdefault}{\color[rgb]{0,0,0}$p_1 > q_2$.}%
}}}}
\put(6526,-830){\makebox(0,0)[rb]{\smash{{\SetFigFont{8}{9.6}{\sfdefault}{\mddefault}{\updefault}{\color[rgb]{0,.56,.56}$(p_1,q_1)$}%
}}}}
\put(8056,-826){\makebox(0,0)[lb]{\smash{{\SetFigFont{8}{9.6}{\familydefault}{\mddefault}{\updefault}{\color[rgb]{0,0,0}$s_2$;}%
}}}}
\put(6616,-1951){\makebox(0,0)[lb]{\smash{{\SetFigFont{8}{9.6}{\rmdefault}{\mddefault}{\itdefault}{\color[rgb]{0,0,0}No task with max prio $\geq p_1$ suspends $A$.}%
}}}}
\put(1049,-830){\makebox(0,0)[rb]{\smash{{\SetFigFont{8}{9.6}{\sfdefault}{\mddefault}{\updefault}{\color[rgb]{0,.56,.56}$(p_1,q_1)$}%
}}}}
\put(1306,-1771){\makebox(0,0)[lb]{\smash{{\SetFigFont{8}{9.6}{\rmdefault}{\mddefault}{\itdefault}{\color[rgb]{0,0,0}No task with max prio $\geq p_1$ resumes $B$.}%
}}}}
\put(1306,-4606){\makebox(0,0)[lb]{\smash{{\SetFigFont{8}{9.6}{\rmdefault}{\mddefault}{\itdefault}{\color[rgb]{0,0,0}Sched suspended or no task with prio $\geq p_1$ resets flag.}%
}}}}
\put(1049,-3665){\makebox(0,0)[rb]{\smash{{\SetFigFont{8}{9.6}{\sfdefault}{\mddefault}{\updefault}{\color[rgb]{0,.56,.56}$(p_1,q_1)$}%
}}}}
\put(1306,-4786){\makebox(0,0)[lb]{\smash{{\SetFigFont{8}{9.6}{\rmdefault}{\mddefault}{\itdefault}{\color[rgb]{0,0,0}No block before $s_1$ or no other task resets flag.}%
}}}}
\put(4546,-3661){\makebox(0,0)[lb]{\smash{{\SetFigFont{8}{9.6}{\sfdefault}{\mddefault}{\updefault}{\color[rgb]{0,.56,.56}$(p_2,q_2)$}%
}}}}
\put(1466,-3661){\makebox(0,0)[lb]{\smash{{\SetFigFont{8}{9.6}{\familydefault}{\mddefault}{\updefault}{\color[rgb]{0,0,0}$s_1$;}%
}}}}
\put(2836,-3653){\makebox(0,0)[b]{\smash{{\SetFigFont{8}{9.6}{\familydefault}{\mddefault}{\updefault}{\color[rgb]{0,0,0}$\noibsymb$}%
}}}}
\put(1306,-4966){\makebox(0,0)[lb]{\smash{{\SetFigFont{8}{9.6}{\rmdefault}{\mddefault}{\itdefault}{\color[rgb]{0,0,0}$q_1 < p_2$}%
}}}}
\put(3241,-3661){\makebox(0,0)[lb]{\smash{{\SetFigFont{8}{9.6}{\familydefault}{\mddefault}{\updefault}{\color[rgb]{0,0,0}$s_2$;}%
}}}}
\put(6711,-4111){\makebox(0,0)[lb]{\smash{{\SetFigFont{7}{8.4}{\ttdefault}{\mddefault}{\updefault}{\color[rgb]{0,0,0}\unlock(l);}%
}}}}
\put(8076,-3652){\makebox(0,0)[b]{\smash{{\SetFigFont{8}{9.6}{\familydefault}{\mddefault}{\updefault}{\color[rgb]{0,0,0}$\noibsymb$}%
}}}}
\put(8346,-3185){\makebox(0,0)[lb]{\smash{{\SetFigFont{7}{8.4}{\ttdefault}{\mddefault}{\updefault}{\color[rgb]{0,0,0}\lock(l);}%
}}}}
\put(6696,-3185){\makebox(0,0)[lb]{\smash{{\SetFigFont{7}{8.4}{\ttdefault}{\mddefault}{\updefault}{\color[rgb]{0,0,0}\lock(l);}%
}}}}
\put(8346,-4110){\makebox(0,0)[lb]{\smash{{\SetFigFont{7}{8.4}{\ttdefault}{\mddefault}{\updefault}{\color[rgb]{0,0,0}\unlock(l);}%
}}}}
\put(8356,-3660){\makebox(0,0)[lb]{\smash{{\SetFigFont{8}{9.6}{\familydefault}{\mddefault}{\updefault}{\color[rgb]{0,0,0}$s_2$;}%
}}}}
\put(6726,-3660){\makebox(0,0)[lb]{\smash{{\SetFigFont{8}{9.6}{\familydefault}{\mddefault}{\updefault}{\color[rgb]{0,0,0}$s_1$;}%
}}}}
\put(6752,-7091){\makebox(0,0)[lb]{\smash{{\SetFigFont{7}{8.4}{\ttdefault}{\mddefault}{\updefault}{\color[rgb]{0,0,0}\resumesched}%
}}}}
\put(6737,-6165){\makebox(0,0)[lb]{\smash{{\SetFigFont{7}{8.4}{\ttdefault}{\mddefault}{\updefault}{\color[rgb]{0,0,0}\suspendsched;}%
}}}}
\put(6747,-6640){\makebox(0,0)[lb]{\smash{{\SetFigFont{8}{9.6}{\familydefault}{\mddefault}{\updefault}{\color[rgb]{0,0,0}$s_1$;}%
}}}}
\put(8117,-6632){\makebox(0,0)[b]{\smash{{\SetFigFont{8}{9.6}{\familydefault}{\mddefault}{\updefault}{\color[rgb]{0,0,0}$\noibsymb$}%
}}}}
\put(8387,-6640){\makebox(0,0)[lb]{\smash{{\SetFigFont{8}{9.6}{\familydefault}{\mddefault}{\updefault}{\color[rgb]{0,0,0}$s_2$;}%
}}}}
\put(1471,-7112){\makebox(0,0)[lb]{\smash{{\SetFigFont{7}{8.4}{\ttdefault}{\mddefault}{\updefault}{\color[rgb]{0,0,0}\enable;}%
}}}}
\put(2836,-6653){\makebox(0,0)[b]{\smash{{\SetFigFont{8}{9.6}{\familydefault}{\mddefault}{\updefault}{\color[rgb]{0,0,0}$\noibsymb$}%
}}}}
\put(1486,-6661){\makebox(0,0)[lb]{\smash{{\SetFigFont{8}{9.6}{\familydefault}{\mddefault}{\updefault}{\color[rgb]{0,0,0}$s_1$;}%
}}}}
\put(3147,-6640){\makebox(0,0)[lb]{\smash{{\SetFigFont{8}{9.6}{\familydefault}{\mddefault}{\updefault}{\color[rgb]{0,0,0}$s_2$;}%
}}}}
\put(1456,-6186){\makebox(0,0)[lb]{\smash{{\SetFigFont{7}{8.4}{\ttdefault}{\mddefault}{\updefault}{\color[rgb]{0,0,0}\disable;}%
}}}}
\put(1471,-1277){\makebox(0,0)[lb]{\smash{{\SetFigFont{7}{8.4}{\ttdefault}{\mddefault}{\updefault}{\color[rgb]{0,0,0}\resume(B);}%
}}}}
\put(1456,-351){\makebox(0,0)[lb]{\smash{{\SetFigFont{7}{8.4}{\ttdefault}{\mddefault}{\updefault}{\color[rgb]{0,0,0}\suspend$(B)$;}%
}}}}
\put(1466,-826){\makebox(0,0)[lb]{\smash{{\SetFigFont{8}{9.6}{\familydefault}{\mddefault}{\updefault}{\color[rgb]{0,0,0}$s_1$;}%
}}}}
\put(1306,-1956){\makebox(0,0)[lb]{\smash{{\SetFigFont{8}{9.6}{\rmdefault}{\mddefault}{\itdefault}{\color[rgb]{0,0,0}No block before $s_1$ or no other task resumes $B$.}%
}}}}
\put(2836,-818){\makebox(0,0)[b]{\smash{{\SetFigFont{8}{9.6}{\familydefault}{\mddefault}{\updefault}{\color[rgb]{0,0,0}$\noibsymb$}%
}}}}
\put(3140,-826){\makebox(0,0)[lb]{\smash{{\SetFigFont{8}{9.6}{\familydefault}{\mddefault}{\updefault}{\color[rgb]{0,0,0}$s_2$;}%
}}}}
\end{picture}%

\caption{Rules that guarantee $s_2$ cannot occur in between $s_1$
  (i.e.\@ $s_1 \noibsymb s_2$)}
\label{fig:occur-in-between}
\end{figure}




\section{Implementation and Evaluation}
\label{sec:implementation}



We have implemented our analysis for FreeRTOS \cite{freertos}
applications in a tool called \tool\ (for ``RTOS App Racer'').
Our IDA language is closely modelled on the syntax and semantics of
FreeRTOS apps, and hence we continue to use the IDA commands in place of
the FreeRTOS commands in this section.
Our analysis is implemented in the CIL static analysis framework
\cite{cil} using OCaml, and comprises a few 
pre-analyses followed by the main analysis for checking the
conditions.
For convenience we assume that each \create\ statement uses a different
thread identifier.

\emph{Priority Analysis.} The priority analysis determines the min
and max dynamic priority at each statement in each thread
function. This is done in 2 passes as follows.
We first perform an interval-based analysis for each function $A$,
maintaining an interval $[p,q]$ of possible priority values at each
statement.
The initial interval is $[p_0,q_0]$, given by the min and max
priorities among threads that run $A$.
The transfer function for a \cmd{set\_priority}$(\mathtt{NULL},p)$ statement
returns the interval $[p,p]$, and is identity for
other statements.
The join is the standard join on the interval lattice.
In the second pass, for each statement \cmd{set\_priority}$(t,p')$
where $t$ may run $A$, we update the interval $[p,q]$ at each
statement in $A$ to $[\min(p,p'),\max(q,p')]$.


\emph{Suspend/Resume Analysis.} This analysis determines the set of
tasks which can suspend or resume a given task. We maintain a set of
task functions $\susplist(A)$ and $\reslist(A)$ for each task
function $A$, containing the set of tasks that can suspend and resume
$A$ respectively. For each task $B$ with a
\suspend$(A)$ or \resume$(A)$ statement, we add $B$ to
$\susplist(A)$ or $\reslist(A)$ respectively.
    
\emph{Lockset Analysis.} A standard lockset analysis aims to compute
the set of locks that are \emph{must} held at each program point.
On a \lock$(l)$ statement, the transfer function adds $l$ to the set
of locks held after this statement, while
for an \unlock($l$) statement we remove $l$ from the set of locks held.
The join operation is simply the intersection of the locksets at the
incoming points.
In our setting, apart from the standard locks, we use notional locks
that correspond to the different kinds of code blocks used in the
rules of Fig.~\ref{fig:occur-in-between}.
Thus, for each \suspend$(A)$-\resume$(A)$ block (rule~C1) we use a
notional lock $S_A$ that we ``acquire'' at the \suspend$(A)$ statement
and ``release'' at the \resume$(A)$ statement.
The lockset analysis would then let us identify a
\suspend$(A)$-\resume$(A)$ block by the fact that the lock $S_A$ is
held throughout these statements.
In a similar way we use locks $F_{\mathit{flag}}^{\mathit{set}}$ and
$F_{\mathit{flag}}^{\mathit{chk}}$ for each flag variable
$\mathit{flag}$, corresponding to the two blocks in rule~C3; lock
$D$ for \disable-\enable\ (rule~C5); and lock $S$ for
\suspendsched-\resumesched\ (rule~C6).

\begin{figure}
    \centering
     \includegraphics[scale=0.4]{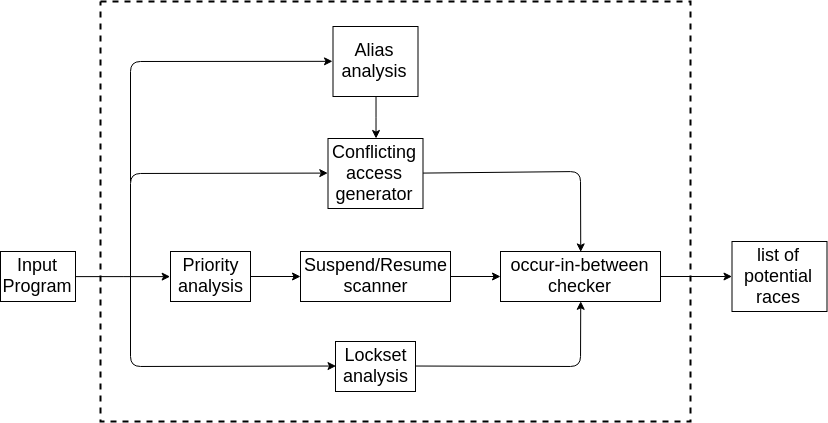}
    \caption{Architecture of \tool}
    \label{fig:90}
\end{figure}

\emph{Main Analysis.}
The overall analysis computes conflicting accesses by scanning for global
variables having shared accesses in different
threads with at least one access being a write access.
We use CIL's inbuilt alias analysis to resolve pointers to shared global
variables.
The information
obtained from priority and lockset analysis is used to check for the
cannot occur-in-between relation between the pair of conflicting
accesses, using rules C1--C6. If
both accesses in the pair cannot occur-in-between each other,
they are declared to be non-racy; else they are declared to be potentially
racy.
A schematic of our implementation is shown in Fig~\ref{fig:90}.

Finally, the analysis allows a couple of command-line switches to
handle some of the configuration options of FreeRTOS.
Certain locks (called ``mutex'' locks) have a priority inheritance
mechanism associated with them: when a higher priority thread is 
waiting on a mutex already acquired by a lower priority thread, the
lower priority thread has its priority bumped up to the priority of
the higher priority thread.
Anticipating a need while translating nxtOSEK applications, we also
allow a ceiling priority mechanism for mutexes which immediately increases
the priority of the acquiring thread to the max priority of all threads
that might ever acquire the mutex.
To handle this we adapt the transfer function of our priority analysis
for a \lock$(l)$ statement, when $l$ is a mutex, to return
$[p,\max(q,p')]$ in the case of priority inheritance, and
$[p',p']$ in the case of ceiling priority, where $[p,q]$ is the
incoming priority interval and $p'$ is the max priority of any task
that might acquire $l$.
We also
provide a switch to disallow round-robin scheduling within threads of
the same priority, and handle it by modifying the cannot
occur-in-between conditions for C1 and C2 by replacing ``$>$'' by ``$\geq$''
for the dynamic priorities.

\emph{Experimental Evaluation.}
We tested our analysis on 12 FreeRTOS applications, shown in
Tab.~\ref{tab:exps}.
The first 9 are nxtOSEK test programs \cite{nxtosek} analysed in
\cite{SchwarzSVLM11}, which were converted to FreeRTOS programs
taking care to preserve the nxtOSEK semantics which these programs
use.
nxtOSEK uses a priority
ceiling protocol for mutex locks and no round-robin
scheduling between same priority tasks.
The next 2 are demo applications in FreeRTOS and
finally, rangefinder.c is the converted version of
an ArduPilot subsystem \cite{ardupilot} originally in ChibiOS/C++.

The examples used by Schwarz et al. \cite{SchwarzSVLM11}
consist of bipedrobot.c which is part of the control software of a
self-balancing robot, pe\_test.c which tests preemptive scheduling,
res\_test.c which tests resource based synchronization, tt\_test.c
where tasks are time-triggered, usb\_test.c which tests usb
communication, pingpong.c where two tasks set a variable to ``ping''
and ``pong'' using a mutex and counter.c where one task increments the
fields of a structure and the other task resets and prints these
fields. The programs example.c and example\_fun.c are examples from
\cite{SchwarzSVLM11}.
The FreeRTOS demo dynamic.c consists of three tasks which use
different mechanisms to access a shared global counter. IntQueue.c is
another FreeRTOS demo where tasks share global arrays and
counters. 
Finally rangefinder.c is an ArduPilot subsystem with three task threads and
one ISR thread which share access to a state variable and ring and
bounce buffers.

Table~\ref{tab:exps} shows the results of our analysis on these programs.
We ran these experiments on a Intel Core i5-8250U 1.60GHz Quad CPU
machine under Ubuntu 18.04.4. Conf. acc. denotes the total
number of pairs of conflicting accesses to shared global variables in
the program. True races denotes the number of actual races existing
in the program. \tool\ ``Pot.\@ Races'' denotes the number of conflicting
accesses flagged to be potentially racy by the analysis. \%Elim.\ denotes the fraction of conflicting accesses declared to be
non-racy and \%Prec.\ denotes the fraction of potential races
which are actual races. Pot.\@ Races from \cite{SchwarzSVLM11} and
\cite{ChopraPD19} denote the number of potential races flagged using 
their respective techniques.

\begin{table}
  \centering
  \begin{small}
    \begin{tabular}{|l|r|r|r|r|r|r|r|r|r|}
\hline

\multirow{3}{4em}{Program}& \multirow{3}{2em}{LoC} & \multirow{3}{2em}{Conf. acc.} & \multirow{3}{2em}{True Races}  & \multicolumn{4}{c|}{\tool} & Pot. & Pot. \\\cline{5-8} 
&&&& Time & Pot. & \% & \% & Races & Races \\ 
&&&& (in s)    & Races    &Elim.   & Prec.   & \cite{SchwarzSVLM11} &  \cite{ChopraPD19} \\ \hline
 bipedrobot.c   & 338 & 3     & 0     & 1.39       & 2        & 33.33     & 0.00      & 2 & 10\\
 pe\_test.c     & 95  & 4     & 3     & 0.01       & 3        & 25.00     & 100.00    & 3 & 7 \\
 res\_test.c    & 110 & 40    & 8     & 0.03       & 9        & 77.50     & 88.88     & 9 & 11\\
 tt\_test.c     & 105 & 5     & 3     & 0.01       & 3        & 40.00     & 100.00    & 3 & 6\\
 usb\_test.c    & 169 & 0     & 0     & 0.02       & 0        & 0.00      & 100.00    & 0 & 52\\
 example.c      & 87  & 13    & 1     & 0.03       & 1        & 92.30     & 100.00    & 1 & 61\\
 example\_fun.c & 102 & 13    & 1     & 0.05       & 4        & 69.23     & 25.00     & 1 & 61\\
 pingpong.c     & 112 & 3     & 0     & 0.01       & 0        & 100.00    & 100.00    & 0 & 7\\
 counter.c      & 96  & 6     & 6     & 0.01       & 6        & 0.00      & 100.00    & 6 & 9\\ 
 dynamic.c      & 429 & 20    & 2     & 0.13       & 6        & 70.00     & 33.33     & 16* & 23\\
 IntQueue.c     & 747 & 42    & 5     & 0.97       & 16       & 61.90     & 31.25     & 10* & 24\\
rangefinder.c   & 394 & 16    & 10    & 0.23       & 10       & 37.50     & 100.00    & 16* & 18\\ \hline
\end{tabular}
  \end{small}
    \caption{Experimental results}
    \label{tab:exps}
\end{table}
  
 In bipedrobot.c, the Task\_Init only runs once and hence the potential races are false positives. The decrement to digits in LowTask races with the read and write access in HighTask in pe\_test.c. In res\_test.c, the read accesses to digits are unprotected due to which it can be an actual race. The decrement of digits in LowTask is unprotected from HighTask and hence it is racy in tt\_test.c. In usb\_test.c, there are no shared accesses between tasks and hence it is trivially race-free. The races in example.c and example\_fun.c are shown in \cite{SchwarzSVLM11}. The ping and pong tasks use a mutex to access the shared variable in pingpong.c and it is race-free. In counter.c, the fields of the global structure are accessed without any protection and hence race with each other. The first initialization of the counter by the controller task in dynamic.c is an actual race with the increment in the continuous increment task because both are created at the same priority and the continuous increment task can preempt the controller when it is initializing the counter. In Intqueue.c some accesses to the shared arrays are real races. In rangefinder.c, the I2C bus thread's access to the state variable is not protected from the main thread and the main thread's access to the ring buffer is not protected from the UART thread which results in a high number of actual races.
 
 The potential races from \cite{SchwarzSVLM11} value is obtained by
 manually estimating the working of the idea in
 \cite{SchwarzSVLM11}. This is marked with a * for the last 3 programs
 as their technique does not handle constructs like dynamically
 changing the priority of a task and hence is potentially unsound for
 these programs. It also results in more false positives for the
 dynamic.c and IntQueue.c examples as protection from synchronization
 mechanisms like suspending another task, disabling interrupts and
 suspending the scheduler is not considered. The number of potential races
 from \cite{ChopraPD19} is obtained using their tool. The tool does
 not consider priorities for synchronization. Moreover it considers each 
 task function to be run by multiple threads even if only one thread runs it
 in the application. These factors add to its imprecision.

In dynamic.c the conflicts in the continuous increment task and the
limited increment task seem to be racy because they occur at the same
priority but the controller task actually ensures that these two can
never be in the ready state at the same time keeping one of these two
suspended at all times. But this is unknown to the analysis when it
encounters the conflicts as this dynamic information about the
controller task cannot be made available at these points. This is the
reason behind the false positives.
The analysis and the test programs with the
results can be found in the repository
\url{bitbucket.org/rishi2289/static\_race\_detect/}.
 

\section{Related Work}
\label{sec:related}
We discuss related work grouped according to the three categories below.

\emph{Static Race Detection.}
The most closely related work is that of Schwarz \etal\
\cite{schwarz2014,SchwarzSVLM11} and Chopra \etal\ \cite{ChopraPD19}.
In \cite{schwarz2014,SchwarzSVLM11} Schwarz \etal\ provide a precise
interprocedural data-flow analysis for checking races in OSEK-like applications
that use the priority ceiling semantics.
Chopra \etal\ \cite{ChopraPD19} propose the notion of disjoint-blocks to detect
data races and carry out data-flow analysis for FreeRTOS-like interrupt-driven
\emph{kernel APIs}. 
In contrast to both these works, our work handles a comprehensive
variety of synchronization mechanisms, including suspend-resume
and setting priorities dynamically. In addition we handle dynamic
thread creation which both these works do not.

In other work in this category Chen \etal\ \cite{Chen2011} develop a
static analysis tool for race detection 
in binaries of interrupt-driven programs with interrupt 
priorities of upto two levels. The tool considers only disable-enable
of interrupts as a synchronization mechanism and does not consider 
interleavings of task threads.
%
Regehr and Cooprider \cite{RegehrC07} describe a source-to-source translation
of an interrupt-driven program to a standard multi-threaded program, and analyze
the translated program for data races. However their translation is inadequate
in our setting and we refer the reader to \cite{ChopraPD19} for the inherent 
problems with such an approach.
Sung \etal\ \cite{SungKW17} propose a modular technique for static verification
of interrupt-driven programs with nesting and priorities.
However, the algorithm does not consider interrupt-related
synchronization mechanisms nor does it consider interleavings of
task threads or interaction with the ISRs.
Wang \etal\ \cite{Wang2017} present SDRacer, an automated framework that detects
and validates race conditions in interrupt-driven embedded 
software. The tool combines static analysis, symbolic execution, and dynamic
simulation.
However, it is unsound as their static analysis does not 
iterate to fixpoint. 
Mine \etal\ \cite{mine2016} extend Astree by employing a thread-modular static
analyzer to soundly report 
data races in embedded C programs with mutex locks and dynamic priorities.
However they do not consider interrupts and synchronization mechanisms
like flag-based and suspend-resume.
Finally, several papers do lockset-based static analysis for data
races in classical concurrent programs
\cite{Sterling93,Engler2003a,VoungJL07,AbadiFF06}. Flanagan \etal\ \cite{FlanaganF00,FlanaganF01} uses type system to track the lockset at each
program point.
However none of these techniques
apply to interrupt-driven programs with non-standard synchronization
mechanisms and switching semantics.



\emph{Model-Checking.}
Several researchers have used model-checking tools like Slam, Blast, and Spin
to precisely model various kinds of synchronization mechanisms
and detect errors exhaustively
\cite{HenzingerJM04,ElmasQT2008,HavelundS99,HavelundLP01,ZengSLH12,AlurMP00,MukherjeeKD17}.
These technique cannot handle dynamic thread creation, and even with a
small bound on the number of threads suffer from state-space explosion.
Liang \etal\ \cite{Liang2018} present an effective method to verify
interrupt-driven software with nested interrupts, based on symbolic execution. 
The method translates a concurrent program into atomic memory read/write 
\emph{events}, and then describe the interleavings of these events as a 
symbolic partial order expressed by a SAT/SMT formula. It is able to verify
only a bounded number of interrupts.

\emph{High-Level Race Detection.}
A ``high-level'' race occurs when two blocks of code representing
critical accesses overlap in an execution.
Our definition of a data race between statements $s_1$ and $s_2$ in
program $P$ can thus be phrased as a
high-level race on the \myskip-blocks in the augmented program $P_{s_1,s_2}$.
Artho \etal\ \cite{Artho03high}, %
von Praun and Gross \cite{PraunG04}, and Pessanha
\etal\ \cite{DiasPL12}
study a ``view''-based notion of high-level races and carry out
lockset based static analysis to detect high-level races.
Singh \etal\ \cite{SinghPDD19} use the disjoint-block notion of
\cite{ChopraPD19}
to detect high-level races in several RTOS \emph{kernels}. They
consider some non-standard 
synchronization mechanisms and also the relative scheduling priorities of 
specialized threads like callbacks and software interrupts.
However none of these techniques handle the full gamut of synchronization
mechanisms we address, and hence would be very imprecise for our
applications.

\section{Conclusions and Future Work}
\label{sec:conclusion}

We have presented an efficient and precise way to detect
data-races in RTOS applications that use a variety of non-standard
synchronization constructs and idioms.
Going forward we would like to extend our tool to be able to handle
large real-life applications like ArduPilot which are written in C++
and run on the ChibiOS RTOS.
We would also like to extend our technique to identify disjoint-block patterns
so that we can carry out efficient data-flow analysis
\cite{ChopraPD19} for such applications.

\bibliography{references}
\appendix
\section{Semantics}\label{app:sem}
\scalebox{0.9}{
\[
 \begin{prooftree}
   \Hypo{c=\myskip \indent t=r \indent \pc(t)=l}
   
   \Infer1[\tiny SKIP]{s \trel_{\instr} \langle \blocked, \suspended, \ready, \pmap, \acquired, \fmap, \pc[t \mapsto l'], \env, r, i, \sd, \id \rangle}
 \end{prooftree}
\]
}
\\[0.5cm]
\scalebox{0.9}{
\[
 \begin{prooftree}
   \Hypo{c=\myskip \indent t \in \ready \indent \isr{t} \indent t \neq r \indent \pc(t)=l=\ent_{\fmap(t)} \indent \pmap(t) > \pmap(r) \indent \id = \false}
   
   \Infer1[\tiny SKIP-INT]{s \trel_{\instr} \langle \blocked, \suspended, \ready, \pmap, \acquired, \fmap, \pc[t \mapsto l'], \env, t, r, \sd, \id \rangle}
 \end{prooftree}
\]
}
\\[0.5cm]
\scalebox{0.9}{
\[
 \begin{prooftree}
  \Hypo{c=\var{x}:=\var{e} \indent t=r \indent \pc(t)=l}
  
  \Infer1[\tiny ASSIGN]{s \trel_{\instr} \langle \blocked, \suspended, \ready, \pmap, \acquired, \fmap, \pc[t \mapsto l'], \env[x \mapsto \llbracket e \rrbracket_{\env} ], r, i, \sd, \id \rangle}
 \end{prooftree}
\]
}
\\[0.5cm]
\scalebox{0.9}{
\[
 \begin{prooftree}
  \Hypo{c=\var{x}:=\var{e} \indent t \in \ready \indent \isr{t} \indent t \neq r \indent \pc(t)=l=\ent_{\fmap(t)} \indent \pmap(t) > \pmap(r) \indent \id = \false}
  
  \Infer1[\tiny ASSIGN-INT]{s \trel_{\instr} \langle \blocked, \suspended, \ready, \pmap, \acquired, \fmap, \pc[t \mapsto l'], \env[x \mapsto \llbracket e \rrbracket_{\env} ], t, r, \sd, \id \rangle}
 \end{prooftree}
\]
}
\\[0.5cm]
\scalebox{0.9}{
\[
 \begin{prooftree}
  \Hypo{c=\assume(\var{b}) \indent t=r \indent \pc(t)=l \indent \llbracket b \rrbracket_{\env}=\true }
  
  \Infer1[\tiny ASSUME]{s \trel_{\instr} \langle \blocked, \suspended, \ready, \pmap, \acquired, \fmap, \pc[t \mapsto l'], \env, r, i, \sd, \id \rangle}
 \end{prooftree}
\]
}
\\[0.5cm]
\scalebox{0.9}{
\[
 \begin{prooftree}
  \Hypo{c=\assume(\var{b}) \smspace t \in \ready \smspace \isr{t} \smspace t \neq r \smspace \pc(t)=l=\ent_{\fmap(t)} \smspace \pmap(t) > \pmap(r) \smspace \llbracket b \rrbracket_{\env}=\true  \smspace \id = \false}
  
  \Infer1[\tiny ASSUME-INT]{s \trel_{\instr} \langle \blocked, \suspended, \ready, \pmap, \acquired, \fmap, \pc[t \mapsto l'], \env, t, r, \sd, \id \rangle}
 \end{prooftree}
\]
}
\\[0.5cm]
%

%

\scalebox{0.9}{
\[
 \begin{prooftree}
  \Hypo{ { c=\create(\var{A}, \var{p}, \var{v}) \smspace t=r \smspace
      \task{t} \smspace A \in \codes \smspace \type(A) = \taskfn \smspace ts \notin \threads \smspace (p \leq \pmap(r)  \vee  (\sd \vee \id) = \true) }} 
  
  \Infer1[\tiny CREATE-NS]{s \trel_{\instr} \langle \blocked,
    \suspended, \ready \union \{ts\}, \pmap[ts \mapsto p], \acquired,
    \fmap[ts \mapsto A], \pc[t \mapsto l', ts \mapsto \ent_{A}], \env[v \mapsto ts], r, i, \sd, \id \rangle}
 \end{prooftree}
\]
}
\\[0.5cm]
%

\scalebox{0.9}{
\[
 \begin{prooftree}
  \Hypo{{ c=\create(\var{A}, \var{p}, \var{v}) \smspace t=r \smspace \task{t} \smspace A \in \codes \smspace \type(A) = \taskfn \smspace  ts \notin \threads \smspace p > \pmap(r) \smspace (\sd \vee \id)=\false  }}
  
  \Infer1[\tiny CREATE-CS]{s \trel_{\instr} \langle \blocked, \suspended, \ready \union \{ts\}, \pmap[ts \mapsto p], \acquired, \fmap[ts \mapsto A], \pc[t \mapsto l', ts \mapsto \ent_{A}], \env[\var{v} \mapsto ts], ts, i, \sd, \id \rangle}
 \end{prooftree}
\]
}
\\[0.5cm]
\scalebox{0.8}{
\[
\hspace*{-0.5cm}
 \begin{prooftree}
  \Hypo{ { c=\cmd{set\_priority}(\var{ts}, \var{p}) \smspace t=r \smspace \task{t} \smspace \pc(t)=l \smspace p \in \nat \smspace \task{ts} \smspace ts \in \threads \smspace (\ (\pmap(r) \geq p ) \ \vee \ (\pmap(r) < p \ \wedge \ ts \in (\blocked \union \suspended)) \ \vee \ (\sd \vee \id)=\true \ ) } }
  
  \Infer1[\tiny SETP-NS]{s \trel_{\instr} \langle \blocked, \suspended, \ready, \pmap[ts \mapsto p], \acquired, \fmap, \pc[t \mapsto     l'], \env, r, i, \sd, \id \rangle}
 \end{prooftree}
\]
}
\\[0.5cm]
\scalebox{0.9}{
\[
 \begin{prooftree}
  \Hypo{ { c=\cmd{set\_priority}(\var{ts}, \var{p}) \smspace t=r \smspace \task{t} \smspace \pc(t)=l \smspace p \in \nat \smspace \task{ts} \smspace ts \in \ready \smspace p > \pmap(r) \smspace (\sd  \vee  \id)=\false } }
  
  \Infer1[\tiny SETP-CS]{s \trel_{\instr} \langle \blocked, \suspended, \ready, \pmap[ts \mapsto p], \acquired, \fmap, \pc[t \mapsto     l'], \env, ts, i, \sd, \id \rangle}
 \end{prooftree}
\]
}
\\[0.5cm]
\scalebox{0.9}{
\[
 \begin{prooftree}
  \Hypo{c=\suspend(\var{ts}) \indent \task{t} \indent t = r \neq ts \indent ts \in \threads \indent \pc(t)=l }
  
  \Infer1[\tiny SUS-NS]{s \trel_{\instr} \langle \blocked - \{ts\}, \suspended \union \{ts\}, \ready - \{ts\}, \pmap, \acquired, \fmap, \pc[t \mapsto l'], \env, r, i, \sd, \id \rangle}
 \end{prooftree}
\]
}
\\[0.5cm]
\scalebox{0.8}{
\[
\hspace*{-1cm}
 \begin{prooftree}
  \Hypo{ { c=\suspend(\var{ts}) \smspace \task{t} \smspace t=r=ts \smspace \pc(t)=l \smspace (\sd \vee \id)=\false \smspace \exists ts' \in \ready. \task{ts'} \ \wedge \ ts' \neq r \ \wedge \ \pmap(ts') = max(\{\pmap(u) | u \in \ready -\{r\} \ \wedge\ \task{u} \})}}
  
  \Infer1[\tiny SUS-CS]{s \trel_{\instr} \langle \blocked, \suspended \union \{r\}, \ready-\{r\}, \pmap, \acquired, \fmap, \pc[t \mapsto l'], \env, ts', i, \sd, \id \rangle}
 \end{prooftree}
\]
}
\\[0.5cm]
\scalebox{0.9}{
\[
 \begin{prooftree}
  \Hypo{c=\resume(\var{ts}) \smspace \task{t} \smspace t = r \neq ts \smspace \pc(t)=l \smspace ts \in (\suspended \union \ready) \smspace ( (\sd \ \vee \ \id)=\true \ \vee \ \pmap(r) \geq \pmap(ts)  ) }
  
  \Infer1[\tiny RES-NS]{s \trel_{\instr} \langle \blocked, \suspended - \{ts\}, \ready \union \{ts\}, \pmap, \acquired, \fmap, \pc[t \mapsto l'], \env, r, i, \sd, \id \rangle}
 \end{prooftree}
\]
}
\\[0.5cm]
\scalebox{0.9}{
\[
 \begin{prooftree}
  \Hypo{c=\resume(\var{ts}) \smspace \task{t} \smspace t = r \neq ts \smspace \pc(t)=l \smspace ts \in \suspended \smspace (\sd \ \vee \ \id)=\false \smspace \pmap(ts) > \pmap(r) }
  
  \Infer1[\tiny RES-CS]{s \trel_{\instr} \langle \blocked, \suspended - \{ts\}, \ready \union \{ts\}, \pmap, \acquired, \fmap, \pc[t \mapsto l'], \env, ts, i, \sd, \id \rangle}
 \end{prooftree}
\]
}
\\[0.5cm]
\scalebox{0.9}{
\[
 \begin{prooftree}
  \Hypo{c=\suspendsched \indent t=r \indent \task{t}\indent \pc(t)=l } 
  
  \Infer1[\tiny SUSSCH]{s \trel_{\instr} \langle \blocked, \suspended, \ready, \pmap, \acquired, \fmap, \pc[t \mapsto l'], \env, r, i, \true, \id \rangle}
 \end{prooftree}
\]
}
\\[0.5cm]
\scalebox{0.9}{
\[
 \begin{prooftree}
  \Hypo{c=\resumesched \indent t=r \indent \task{t}\indent \pc(t)=l \indent (\forall ts \in \ready. \task{ts} \wedge \pmap(r) \geq \pmap(ts) \ \vee\ \id = \true)}
  
  \Infer1[\tiny RESSCH-NS]{s \trel_{\instr} \langle \blocked, \suspended, \ready, \pmap, \acquired, \fmap, \pc[t \mapsto l'], \env, r, i, \false, \id \rangle}
 \end{prooftree}
\]
}
\\[0.5cm]
\scalebox{0.9}{
\[
 \begin{prooftree}
  \Hypo{{ c=\resumesched \smspace t=r \smspace \task{t} \smspace \pc(t)=l \smspace \exists ts \in \ready. \task{ts} \wedge \pmap(ts) = max(\{\pmap(u) | u \in \ready \ \wedge \ \task{u} \}) \smspace \id = \false }}
  
  \Infer1[\tiny RESSCH-CS]{s \trel_{\instr} \langle \blocked, \suspended, \ready, \pmap, \acquired, \fmap, \pc[t \mapsto l'], \env, ts, i, \false, \id \rangle}
 \end{prooftree}
\]
}
\\[0.5cm]
\scalebox{0.9}{
\[
 \begin{prooftree}
  \Hypo{c=\disable \indent t=r \indent \pc(t)=l }
  
  \Infer1[\tiny DISINT]{s \trel_{\instr} \langle \blocked, \suspended, \ready, \pmap, \acquired, \fmap, \pc[t \mapsto l'], \env, r, i, \sd, \true \rangle}
 \end{prooftree}
\]
}
\\[0.5cm]
\scalebox{0.9}{
\[
 \begin{prooftree}
  \Hypo{c=\disable \indent t \in \ready \indent \isr{t} \indent t \neq r \indent \pc(t)=l=\ent_{\fmap(t)} \indent \pmap(t) > \pmap(r) \indent \id = \false }
  
  \Infer1[\tiny DISINT-INT]{s \trel_{\instr} \langle \blocked, \suspended, \ready, \pmap, \acquired, \fmap, \pc[t \mapsto l'], \env, t, r, \sd, \true \rangle}
 \end{prooftree}
\]
}
\\[0.5cm]
\scalebox{0.9}{
\[
 \begin{prooftree}
  \Hypo{c=\enable \indent t=r \indent \pc(t)=l \indent (\forall ts \in \ready. \task{ts} \wedge \pmap(r) \geq \pmap(ts) \ \vee \ \sd = \true \ \vee \ \isr{r}) }
  
  \Infer1[\tiny ENINT-NS]{s \trel_{\instr} \langle \blocked, \suspended, \ready, \pmap, \acquired, \fmap, \pc[t \mapsto l'], \env, r, i, \sd, \false \rangle}
 \end{prooftree}
\]
}
\\[0.5cm]
\scalebox{0.9}{
\[
 \begin{prooftree}
  \Hypo{ { c=\enable \smspace t=r \smspace \task{t} \smspace \pc(t)=l \smspace \exists ts \in \ready. \task{ts} \wedge \pmap(ts) = max(\{\pmap(u) | u \in \ready \wedge \task{u} \}) \smspace \sd = \false }}
  
  \Infer1[\tiny ENINT-CS]{s \trel_{\instr} \langle \blocked, \suspended, \ready, \pmap, \acquired, \fmap, \pc[t \mapsto l'], \env, ts, i, \sd, \false \rangle}
 \end{prooftree} 
\]
}
\\[0.5cm]
\scalebox{0.9}{
\[
 \begin{prooftree}
  \Hypo{c=\enable \indent t \in \ready \indent \isr{t} \indent t \neq r \indent \pc(t)=l=\ent_{\fmap(t)} \indent \pmap(t) > \pmap(r) \indent \id = \false }
  
  \Infer1[\tiny ENINT-INT]{s \trel_{\instr} \langle \blocked, \suspended, \ready, \pmap, \acquired, \fmap, \pc[t \mapsto l'], \env, t, r, \sd, \false \rangle}
 \end{prooftree}
\]
}
\\[0.5cm]
\scalebox{0.9}{
\[
 \begin{prooftree}
  \Hypo{c=\lock(\var{m}) \indent t=r \indent \pc(t)=l \indent (\acquired(m)=\undef \ \vee \ \acquired(m) = r) }
  
  \Infer1[\tiny LOCK-AQ]{s \trel_{\instr} \langle \blocked, \suspended, \ready, \acquired[m \mapsto r], \fmap, \pc[t \mapsto l'], \env, r, i, \sd, \id \rangle}
 \end{prooftree}
\]
}
\\[0.5cm]
\scalebox{0.8}{
\[
\hspace*{-3cm}
 \begin{prooftree}
  \Hypo{ { c=\lock(\var{m}) \smspace t=r \smspace \task{t} \smspace \pc(t)=l \smspace \acquired(m) = ts \neq r \smspace (\sd \vee \id) = \false \smspace \exists ts' \in \ready. ts' \neq r \ \wedge \ \task{ts'} \ \wedge \ \pmap(ts') = max(\{ \pmap(u) | u \in \ready-\{r\} \wedge \task{u} \}) }}
  
  \Infer1[\tiny LOCK-CS]{s \trel_{\instr} \langle \blocked \union \{r\}, \suspended, \ready - \{r\}, \acquired, \fmap, \pc, \env, ts', i, \sd, \id \rangle}
 \end{prooftree}
\]
}
\\[0.5cm]
\scalebox{0.9}{
\[
 \begin{prooftree}
  \Hypo{c=\lock(\var{m}) \smspace t \in \ready \smspace \isr{t} \smspace t \neq r \smspace \pc(t)=\ent_{\fmap(t)} \smspace \pmap(t) > \pmap(r) \smspace \acquired(m) = \undef \smspace \id = \false}
  
  \Infer1[\tiny LOCK-AQ-INT]{s \trel_{\instr} \langle \blocked, \suspended, \ready, \pmap, \acquired[m \mapsto t], \fmap, \pc[t \mapsto l'], \env, t, r, \sd, \id \rangle }
 \end{prooftree}
\]
}
\\[0.5cm]
\scalebox{0.9}{
\[
 \begin{prooftree}
  \Hypo{c=\unlock(\var{m}) \indent t=r \indent \pc(t)=l \indent (\acquired(m) = r \vee \acquired(m) = \undef) }
  
  \Infer1[\tiny UNLOCK]{s \trel_{\instr} \langle \blocked, \suspended, \ready, \acquired[m \mapsto \undef], \fmap, \pc[t \mapsto l'], \env, r, i, \sd, \id \rangle}
 \end{prooftree}
\]
}
\\[0.5cm]
\scalebox{0.9}{
\[
 \begin{prooftree}
  \Hypo{c=\unlock(\var{m}) \indent t \in \ready \indent \isr{t} \indent t \neq r \indent \pc(t)=l=\ent_{\fmap(t)} \indent \pmap(t) > \pmap(r) \indent \acquired(m) \neq t \indent \id = \false }
  
  \Infer1[\tiny UNLOCK-INT]{s \trel_{\instr} \langle \blocked, \suspended, \ready, \acquired, \fmap, \pc[t \mapsto l'], \env, t, r, \sd, \id \rangle}
 \end{prooftree}
\]
}
\\[0.5cm]
\scalebox{0.9}{
\[
 \begin{prooftree}
  \Hypo{c=\block \indent t=r \indent \task{t} \indent \pc(t)=l \indent (\sd \vee \id) = \true }
  
  \Infer1[\tiny BLK-NS]{s \trel_{\instr} \langle \blocked, \suspended, \ready, \pmap, \acquired, \fmap, \pc[t \mapsto l'], \env, r, i, \sd, \id \rangle}
 \end{prooftree}
\]
}
\\[0.5cm]
\scalebox{0.9}{
\[
\hspace*{-2cm}
 \begin{prooftree}
  \Hypo{{ c=\block \smspace t=r \smspace \task{t} \smspace \pc(t)=l \smspace (\sd \vee \id) = \false \smspace \exists ts \in \ready. \task{ts}\ \wedge \ ts \neq r \ \wedge \ \pmap(ts) = max(\{ \pmap(u) | u \in \ready - \{r\} \ \wedge \ \task{u} \}) }}
  
  \Infer1[\tiny BLK-CS]{s \trel_{\instr} \langle \blocked \union \{r\}, \suspended, \ready - \{r\}, \pmap, \acquired, \fmap, \pc[t \mapsto l'], \env, ts, i, \sd, \id \rangle}
 \end{prooftree}
\]
}
\\[0.5cm]
\scalebox{0.9}{
\[
 \begin{prooftree}
  \Hypo{{ c=\start \smspace t=r=0 \smspace  (\sd \vee \id)=\false \smspace \exists ts \in (\suspended \union \ready). \task{ts} \ \wedge \ \pmap(ts) = max(\{\pmap(u) | u \in \suspended \union \ready \ \wedge \ \task{u} \}) }}
  
  \Infer1[\tiny START]{s \trel_{\instr} \langle \blocked, \emptyset, \suspended \union \ready, \pmap, \acquired, \fmap, \pc[t \mapsto l'], \env, ts, i, \false, \false \rangle }
 \end{prooftree}
\]
}
\\[0.5cm]
\scalebox{0.9}{
\[
 \begin{prooftree}
  \Hypo{t \in \blocked \indent \task{r} \indent ((\sd \vee \id)=\true \ \vee \ \pmap(t) \leq \pmap(r) )}
  
  \Infer1[\tiny UNBLK-NS]{s \trel_{*} \langle \blocked - \{t\}, \suspended, \ready \union \{t\}, \pmap, \acquired, \fmap, \pc, \env, r, i, \sd, \id \rangle}
 \end{prooftree}
\]
}
\\[0.5cm]
\scalebox{0.9}{
\[
 \begin{prooftree}
  \Hypo{t \in \blocked \indent \task{r} \indent (\sd \vee \id) = \false \indent \pmap(t) > \pmap(r) }
  
  \Infer1[\tiny UNBLK-CS]{s \trel_{*} \langle \blocked - \{t\}, \suspended, \ready \union \{t\}, \pmap, \acquired, \fmap, \pc, \env, t, i, \sd, \id \rangle}
 \end{prooftree}
\]
}
\\[0.5cm]
\scalebox{0.9}{
\[
 \begin{prooftree}
  \Hypo{t \in \ready \indent \task{t} \indent t \neq r \indent (\sd \vee \id) = \false \indent \pmap(t) = \pmap(r) }
  
  \Infer1[\tiny TSHARE]{s \trel_{*} \langle \blocked, \suspended, \ready, \pmap, \acquired, \fmap, \pc, \env, t, i, \sd, \id \rangle}
 \end{prooftree}
\]
}
\\[0.5cm]
For the commands \myskip, \cmd{x:=e}, \assume, \disable, \enable, \lock, and \unlock\ permitted in an ISR thread, the following constraints need to hold on $s'$. 
If the current ISR thread is executing the last statement then $r'$ is the highest priority ISR which was interrupted, if there exists one, and $i' = i$. If no ISRs were interrupted then $r' = i$, the interrupted task thread and $i' = \main$, a default value. Also, $\pc'(t)=\ent_{\fmap(t)}$.

\end{document}